\documentclass[11pt]{article}
\usepackage[top = 1.5 in, bottom = 1.5 in, left = 1.25 in, right = 1.25 in]{geometry}
\usepackage[all]{xy}
\usepackage[centertags]{amsmath}
\usepackage{amssymb}
\usepackage{amsthm}
\usepackage{color}
\newtheorem{theorem}{Theorem}
\newtheorem{corollary}[theorem]{Corollary}
\newtheorem{lemma}{Lemma}[]
\newtheorem{remark}{Remark}[]

\author{Da-Quan Jiang\thanks{LMAM, School of Mathematical Sciences \& Center for Statistical Science, Peking University, Beijing 100871, P.R. China.} \and Yue Wang\thanks{Department of Applied Mathematics, University of Washington, Seattle, WA 98195, USA. Email address: yuewang@uw.edu}\and Da Zhou\thanks{School of Mathematical Sciences, Xiamen University, Xiamen 361005, P.R. China.}}

\begin{document}

\title{Phenotypic Equilibrium as Probabilistic Convergence in Multi-phenotype Cell Population Dynamics}

\date{}
\maketitle

\begin{abstract}
We consider the cell population dynamics with $n$ different phenotypes. Both the Markovian branching process model (stochastic model) and the ordinary differential equation (ODE) system model (deterministic model) are presented, and exploited to investigate the dynamics of the phenotypic proportions. We will prove that in both models, these proportions will tend to constants regardless of initial population states (``phenotypic equilibrium'') under weak conditions, which explains the experimental phenomenon in Gupta et al.'s paper. We also
prove that Gupta et al.'s explanation is the ODE model under a special assumption. As an application, we will give sufficient and necessary conditions under which the proportion of one phenotype tends to $0$ (die out) or $1$ (dominate). We also extend our results to non-Markovian cases.
\begin{flushleft}
\end{flushleft}
\begin{flushleft}
\end{flushleft}

\end{abstract}

\section{Introduction}
With the same genetic background, cell population may have different cellular phenotypes. This has been one of the major topics in the research of cell population dynamics \cite{AW10,KL05}. Very recently much attention has been paid to the stochastic conversions between different phenotypes \cite{dSdS13,Gupta}. For example, we know that cancer stem cells can give rise to cancer non-stem cells, but cancer non-stem cells can also transform back to cancer stem cells \cite{ZWL13,YQ12}. Generally, we can use a branching process (stochastic model) \cite{Jagers,KA02,YMN,YY09,YY10} or an ODE system (deterministic model) \cite{Murray} to describe the dynamics of  such cell population with multiple phenotypes. However, in many experimental settings, it is difficult or even impossible to count the total cell population \cite{Clayton,YY09,YY10}. Thus in the last fifty years, people began to consider the proportions of cell individuals with distinct phenotypes instead of the absolute numbers of cells of various phenotypes \cite{Jagers}.

We know that through multistep accumulation gene mutations, healthy cells gradually transform to malignant cancer cells, which is the current view of cancer progression \cite{HW00}. Among those mutations, most of them are neutral (``passenger mutations'') and have no effect on cell proliferation. Only a small portion of mutations will bring growth advantage (``driver mutations'') \cite{Bozic}. Since the emerging of mutations can be regarded as purely stochastic, we can model such procedure with multitype branching processes (cf. Bozic et al.'s model \cite{Bozic,KC13}). There have been some results about the emerging time for certain number of driver mutations, and the relation between the number of emerged driver mutations and emerged passenger mutations \cite{Bozic,MK15}. 

In the experiments on breast cancer cell lines, Gupta et al. \cite{Gupta} found that the proportion of each phenotype will tend to a certain constant regardless of the initial population states (``phenotypic equilibrium''). They built a Markovian model, assuming that the evolution of the phenotypic proportions satisfies an $n$-state Markov chain, and used the ergodicity of the Markov chain to explain this phenomenon \cite{Gupta}. However, we find that the Markovian model is just the ODE system model under a special condition. We determine this condition and its biological meaning. Furthermore, we try to remove this condition and explain the experimental phenomenon in \cite{Gupta} under more general context.

In the deterministic model (ODE system), we only consider the average behavior of cell population dynamics (which requires a large initial population). However, using the stochastic model (branching processes), we can study the trajectory behavior. We prove that the proportions will converge not only on average, but also almost surely. This implies that even with a small initial population, we can still observe ``phenotypic equilibrium''. 

In the theory of multitype branching processes, people have observed similar proportion convergence phenomenon and proved such phenomenon in several limit theorems under different conditions \cite{KS67,AN1972,Janson}. Those are possible ways to explain ``phenotypic equilibrium'', but  those required conditions may not be satisfied in experiments. Thus we improve those limit theorems by dropping redundant conditions. We will see that the conditions we need are all biologically reasonable. Therefore, we give a stochastic explanation of ``phenotypic equilibrium''. This result may also be of interests to probabilists.

Generally we only consider Markovian branching processes, but sometimes the biological process is not memoryless, thus we need to consider non-Markovian branching processes. We show that under some conditions, the non-Markovian branching processes can be transformed into Markovian branching processes. Using this trick, we demonstrate similar results for non-Markovian branching processes.

In Section \ref{S2}, we will define some notations, and give the mathematical description of our models, which is based on \cite{KA02} and \cite{Janson}. In Section \ref{S3}, we will describe under which condition the deterministic model becomes the Markovian model in \cite{Gupta}. In Section \ref{S4}, we will prove that under some mild conditions, the ``phenotypic equilibrium'' phenomenon will always happen in the Markovian branching process model. Specifically, we will improve a limit theorem about proportion convergence in multitype branching processes. We will also apply our results to Bozic et al.'s model. In Section \ref{S5}, as an application of our conclusions, we will investigate under what conditions one of the phenotypes will die out or dominate. In Section \ref{S6}, we will show that the above conclusions are still valid in more general cases.

\section{Notations and model description}
\label{S2}
\subsection{Notations}Boldface letter, like $\textbf{A}$, represents matrix. $\textbf{A}'$ means the matrix transpose of $\textbf{A}$. $\textbf{I}$ is the identity matrix. Letter with arrow above, like $\vec{u}$, represents row vectors. $\vec{1}$ and $\vec{0}$ represent all ones and all zeros vectors. Consider the population of cells with $n$ phenotypes: $Y_{1},Y_{2},$ $\cdots{},$$Y_{n}$. In the stochastic model, $\vec{X}(t)=(X_1(t),X_2(t),\cdots, X_n(t))$ is the population at time $t$, where $X_i(t)$ is the population of phenotype $Y_i$. $P_i(t)=X_i(t)/\sum_{i=1}^n X_i(t)$ is the proportion of phenotype $Y_i$, as long as the denominator is not zero. $\vec{P}(t)=(P_1(t),P_2(t),\cdots,P_n(t))$. In the deterministic model, $\vec{x}(t)=(x_{1}(t),x_{2}(t),\cdots{},$ $x_{n}(t))$ is the expected population at time $t$. We only consider the case where each $x_i(t)$ is nonnegative, and at least one of them is positive (to guarantee $|\vec{x}|>0$). $|\vec{x}|=\sum_{i=1}^n x_i(t)$ is the total population. $\vec{p}=\vec{x}/|\vec{x}|$ is the proportions of different subpopulations among the total population.

\subsection{Stochastic model} Assume that the population of cells have $n$ phenotypes: $Y_{1},Y_{2},$ $\cdots{},$$Y_{n}$. Assume that all the cells evolve independently. (During the exponential growth period, this assumption is almost true \cite{ZJ}.) We can present the generalized cell divisions, death and phenotypic conversions as the following reaction form:
$$Y_{i}\stackrel{\alpha_i}{\to}d_{i1}Y_{1}+d_{i2}Y_{2}+\cdots{}+d_{in}Y_{n}.$$
It means that for an $Y_i$ cell, it will live an exponential time (we will consider non-exponential lifetime in Section \ref{S6}) with expectation $1/\alpha_i$ and turn into $d_{i1}$ $Y_{1}$ cells, $d_{i2}$ $Y_{2}$ cells, $\cdots$, $d_{in}$ $Y_{n}$ cells, where $d_{i1},d_{i2},\cdots{},d_{in}$ are random variables taking nonnegative integer values. $d_{i1},d_{i2},$ $\cdots{},d_{in}$ are not necessarily independent, but they are assumed to be independent with the exponential reaction time of any cell. 

For example, an asymmetric division $Y_1\to Y_1+Y_2$ means $(d_{11},d_{12},d_{13},\cdots,d_{1n})=(1,1,0,\cdots,0)$. A conversion $Y_1\to Y_2$ means $(d_{11},d_{12},d_{13},\cdots,d_{1n})=(0,1,0,\cdots,0)$. So the probability distribution on the possible reactions gives the joint probability distribution of $d_{i1},d_{i2},$ $\cdots{},d_{in}$.

In fact this is a multitype continuous-time Markovian branching process $\vec{X}(t)$ with state space $(\mathbb{Z}^\ast)^n$, each component of which represents the population of a phenotype, as defined in \cite{AN1972} and \cite{Janson}. For example, if one $Y_1$ cell splits symmetrically, the process will move from the state $\vec{X}=(s_1,s_2,\cdots{},s_n)$ to the state $\vec{X}=(s_{1}+1,s_2,\cdots{},s_n)$. We require that $\mathbb{E}d_{ij}^2<\infty$, $\forall i,j$. (In experiments, $d_{ij}$ is bounded, thus $\mathbb{E}d_{ij}^2<\infty$ is always true.) Then this process will  always have finite value in finite time (non-explosion) with probability one \cite[Section V.7.1, (3)--(4)]{AN1972}. 

\subsection{Deterministic model} Now we consider the mathematical expectation of the populations, $\vec{x}=\mathbb{E}(\vec{X})$, which is nonnegative. Based on \cite[Section V.7.2, (5)--(9)]{AN1972}, we have the deterministic model, namely the following ODE system:
\begin{equation}
	\mathrm{d}\vec{x}/\mathrm{d}t=\vec{x}\textbf{A}
	\label{1}
\end{equation}
where $\textbf{A}=
\begin{bmatrix}
	a_{1,1} &\cdots{} & a_{1,n}  \\
	\vdots{} & \ddots{} & \vdots{} \\
	a_{n,1} &\cdots{}& a_{n,n}
\end{bmatrix}$, $a_{i,i}=\alpha_i(\mathbb{E}d_{ii}-1)\ge-\alpha_i$, $a_{i,j}=\alpha_j\mathbb{E}d_{ij}\ge0$ $(i\neq{}j)$. 
Define $\vec{b}=\vec{1}\textbf{A}'$. From (\ref{1}), we have $\mathrm{d}|\vec{x}|/\mathrm{d}t=\vec{x}\vec{b}'$. 

\section{The relation between the Markovian model and the deterministic model}
\label{S3}
In the Markovian model \cite{Gupta}, it is assumed that the population proportions $\vec{p}$ satisfies the Kolmogorov forward equations of an $n$-state Markov chain:
\begin{equation}
	\mathrm{d}\vec{p}/\mathrm{d}t=\vec{p}\textbf{Q}
	\label{2}
\end{equation}
where $\textbf{Q}$ is the transition rate matrix, satisfying $\vec{1}\textbf{Q}'=\vec{0}$.
In this section, we will discuss whether such assumption can be satisfied in the deterministic model.

From (\ref{1}), we have
\begin{equation}
	\frac{\mathrm{d}\vec{p}}{\mathrm{d}t}=\frac{\mathrm{d}(\vec{x}/|\vec{x}|)}{\mathrm{d}t}=\frac{|\vec{x}|}{|\vec{x}|^2}\frac{\mathrm{d}\vec{x}}{\mathrm{d}t}-\frac{\vec{x}}{|\vec{x}|^2}\frac{\mathrm{d}|\vec{x}|}{\mathrm{d}t}
	=\frac{\vec{x}\textbf{A}}{|\vec{x}|}-\frac{(\vec{x}\vec{b}')\vec{x}}{|\vec{x}|^2}
	=\vec{p}[\textbf{A}-(\vec{p}\vec{b}')\textbf{I}]
	\label{3}
\end{equation}

If $\vec{b}=k\vec{1}$ for some constant $k$, then (\ref{3}) becomes $\mathrm{d}\vec{p}/\mathrm{d}t=\vec{p}(\textbf{A}-k\textbf{I})$, and $\vec{1}'(\textbf{A}-k\textbf{I})'=\vec{b}'-k\vec{1}'=\vec{0}'$. Thus (\ref{3}) has the same form of (\ref{2}). If $\vec{b}\ne k\vec{1}$, there are non-zero quadratic terms of $p_i(t)$ in (\ref{3}), implying that (\ref{3}) does not have the same form of (\ref{2}).

Notice that $\vec{b}=k\vec{1}$ means
\begin{equation}
	\sum_{i=1}^{n}a_{1,i}=\sum_{i=1}^{n}a_{2,i}=\cdots{}=\sum_{i=1}^{n}a_{n-1,i}=\sum_{i=1}^{n}a_{n,i}(=k).
	\label{4}
\end{equation}
Thus we have
\begin{theorem}
	Equation (\ref{4}) is the sufficient and necessary condition for that the proportions of different phenotypes in the deterministic model (\ref{1}) satisfy the Kolmogorov forward equations of an $n$-state Markov chain.
	\label{8}
\end{theorem}

Now we know that the Markovian model is a special case of the deterministic model. In biology, (\ref{4}) means that the growth rates (average number of descendants produced per unit time) of different phenotypes are the same. This condition might be well satisfied for breast cancer cells, which explains why the data fitting in \cite{Gupta} is satisfactory.

\section{Asymptotic behavior in general cases}
\label{S4}
In general cases, (\ref{4}) is not satisfied since different phenotypes may differ in cell cycling time \cite{PCS,FK08}, then the Markovian model is invalid. Thus we need other methods to study the asymptotic behavior of the population dynamics. In this section, we will prove that under some mild conditions, the proportions of different phenotypes will tend to some constants regardless of initial population states.

From Perron-Frobenius theorem \cite{Seneta,Karlin}, we know that $\bf A$ has a real eigenvalue $\lambda_1$ (called Perron eigenvalue), such that for any eigenvalue $\mu\ne \lambda_1$, Re $\mu<\lambda_1$. $\lambda_1$ has a left eigenvector $\vec{u}$$=(u_1,u_2,\cdots ,u_n)$ (called Perron eigenvector), satisfying $u_i\ge 0, \forall i$ and $\sum_{i=1}^n u_i=1$. When $\lambda_1$ is simple, such $\vec{u}$ is unique. 
We know that the set of all $n$-order real square matrices with repeated eigenvalue has measure $0$ (as a subset of $\mathbb{R}^{n^2}$) \cite{ZWW}. 
Thus it is reasonable to assume that $\lambda_1$ is simple.

\subsection{Deterministic model}
We have proved the following theorem in Appendix B of \cite{ZWW}.
\begin{theorem}
	Assume that $\lambda_1$ is simple. Starting from any initial value except for the point in some zero-measure set, we have $(x_1(t),x_2(t),\cdots,x_n(t))/\exp(\lambda_1 t)$ $\rightarrow$$ c\vec{u}$ as $t\rightarrow \infty$, where $c>0$ is a constant. In this case, the solution of (\ref{4}) will tend to $\vec{u}$ as $t\rightarrow \infty$. Thus (\ref{4}) has one and only one stable fixed point $\vec{u}$ and no stable limit cycle. 
	\label{12}
\end{theorem}
This gives a satisfactory deterministic explanation of the phenotypic equilibrium phenomenon reported in \cite{Gupta}.
\begin{remark}
	If $\lambda_1$ is not simple, then the convergence result may not hold. Consider $\bf A$ with $a_{i,j}=0, \forall i\ne j$ and $a_{i,i}=1,\forall i$. Here $\lambda_1=1$ is not simple, $x_i(t)=x_i(0)\exp(\lambda_1 t)$, and $(x_1(t),x_2(t),\cdots,x_n(t))/\exp(\lambda_1 t)=(x_1(0),x_2(0),\cdots,x_n(0))$ will never change. Convergence to a common point will not occur.
\end{remark}

\subsection{Stochastic model}
\label{S42}
Since 1960s, probabilists proved that for a continuous-time multitype branching process, $\vec{X}(t)/\mathrm{e}^{\lambda_1 t}\rightarrow W\vec{u}$ under different conditions, where $W$ is a nonnegative random variable. In \cite{Ath68}, \cite{AN1972} and \cite{Smythe}, it is required that $\lambda_1>0$ and $\bf A$ is irreducible (this implies $\lambda_1$ is simple). In \cite{AN1972} it is proved that $W=0$ or $W>0$ according to whether the population will become extinct. In \cite{KS67}, it is required that the branching process is discrete in time. In \cite{YY09,YY10} it is required that the initial population tends to infinity. Janson \cite{Janson} requires that $\lambda_1>0$, $\lambda_1$ is simple, and assumes a special condition about communicating classes structure (see Remark \ref{R2}). Based on \cite{Janson} and \cite{AN1972}, we will prove the convergence theorem without Janson's last assumption (Theorem \ref{42}). We can see the benefit of this improvement in Section \ref{S5}.

\subsubsection{Preliminaries}
\label{pre}
In this section, we assume that $\lambda_1$ is simple and positive. $\lambda_1>0$ means that the total cell population is increasing.

Sometimes, the transformation from one phenotype to another phenotype is not reversible. For example, a mature human red blood cell (which loses its nucleus) cannot transform back to a zygote. Thus we need to classify phenotypes according to communicating behaviors. In mathematical language, we need to study communicating classes of $\bf{A}$ when $\bf{A}$ is reducible. 

If a cell of phenotype $Y_i$ can produce (directly or indirectly) a cell of phenotype $Y_j$ and vice versa, then we say phenotype $Y_i$ communicates with phenotype $Y_j$ ($Y_i\leftrightarrow Y_j$). Since ``$\leftrightarrow$'' is an equivalent relation, we can divide the $n$ phenotypes into several disjoint sets (called communicating classes) according to $\bf A$ \cite{Norris}. Then we can order the classes and rearrange the phenotypes suitably to make $\bf A$ block-triangular. (Each diagonal block corresponds to a communicating class.) Thus the eigenvalues of $\bf A$ consist of all eigenvalues of diagonal blocks. Every eigenvalue corresponds to a diagonal block, and then corresponds to a communicating class. (See \cite{Janson} and \cite{KS67} for details.)

Denote the communicating class corresponding to the Perron eigenvalue $\lambda_1$ by $T$.

For example, consider matrix 
$\bf{A}=
\begin{bmatrix}
	\bf{D_1}&\bf{W}& \bf{W}& \bf{0}  \\
	\bf{0}& \bf{D_2}& \bf{0}&\bf{W} \\
	\bf{0}& \bf{0}&\bf{D_3}&\bf{W} \\
	\bf{0}&\bf{0}&\bf{0}&\bf{D_4}
\end{bmatrix}
$, where each $\bf{W}$ represents a different nonnegative matrix (not $\bf 0$). Assume that $\bf D_3$ has the Perron eigenvalue $\lambda_1$, then $\bf D_3$ corresponds to the communicating class $T$. Denote the other three communicating classes by $C_1$, $C_2$, $C_4$.

For two communicating classes $C_i$ and $C_j$, we write $C_i\Rightarrow C_j$ if there exist phenotype $X_{k_i}\in C_i$ and $X_{k_j}\in C_j$ such that $a_{k_j,k_i}>0$. For two communicating classes $C$ and $D$, we write $C\rightarrow D$ if there exist communicating classes $C=C_1,C_2,\cdots{},C_m=D$ such that $C_i\Rightarrow C_{i+1},\forall{}1\leq i<m$. Stipulate that $C_i\Rightarrow C_i $ and $C_i\rightarrow C_i$.

Then we can illustrate the communicating classes in the example above as \\
$
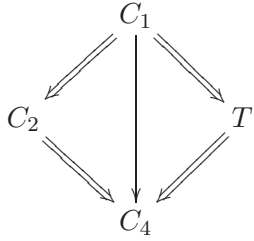
\xymatrix{
	& C_1\ar@{=>}[dl]\ar@{=>}[dr]\ar[dd] &\\
	C_2\ar@{=>}[dr]&     & T \ar@{=>}[dl] \\
	&C_4&       
}$\\

For a communicating class $C$, define $\hat{C}=\{Y_i | Y_i\in C_j , C_j\rightarrow C\}$. In other words, $\hat{C}$ is the set of all phenotypes that can produce (directly or indirectly) phenotypes in $C$. In the example above $\hat{T}=C_1\cup T$.

For a communicating class $C$, define $\bar{C}=\{Y_i | Y_i\in C_j , C\rightarrow C_j\}$. In other words, $\bar{C}$ is the set of all phenotypes that can be produced (directly or indirectly) by phenotypes in $C$. In the example above $\bar{T}=T\cup C_4$.

For the Markovian branching process $\vec{X}(\cdot)$, we say that a cell V with phenotype in $\hat{T}$ becomes ``\textit{essentially extinct}'' if at some time no cell of any phenotypes in $\hat{T}$ is V or its descendants. In other words, V and its descendants become extinct inside $\hat{T}$. We say that a trajectory of the branching process $\vec{X}(\cdot)$ becomes ``\textit{ essentially extinct}'' if at some time no cell of any phenotypes in $\hat{T}$ remains. This means that we can never get a cell with phenotypes in $T$ any more. If so, we cannot have the desired convergence property (Theorem \ref{42}), since the proportions of phenotypes in $T$ should be positive (Lemma \ref{41}). Let the branching process $\vec{X}(\cdot)$ start at any initial population $\vec{X}$$(0)$ as long as it has some cells with phenotypes in $\hat{T}$.

\subsubsection{Results and proofs}
We now state the main result of this paper and then give the proof of it.
\begin{theorem}
	Assume that $\lambda_1$ is simple and positive. Conditioned on essential non-extinction, we have almost surely $(P_1(t),P_2(t),\cdots ,P_n(t))\rightarrow \vec{u}=(u_1,u_2,$ $\cdots ,u_n)$ as $t\rightarrow \infty$.
	\label{42}
\end{theorem}

\begin{lemma}
	Assume that $\lambda_1$ is simple. If for some $i\not=j$, $a_{i,j}>0$ in (\ref{1}), then $u_j>0\Rightarrow{}u_i>0$.
	\label{10}
\end{lemma}
\begin{proof}
	Without loss of generality, let $i=1,j=2$. Assume $u_1=0$, $u_2>0$. Let $(p_1, p_2, \cdots,p_n)=\vec{u}$$=(u_1,u_2,\cdots ,u_n)$ in the first equation of (\ref{3}).  Then it becomes $\mathrm{d}p_1/\mathrm{d}t=\sum_{k=2}^{n}a_{1,k}u_k>0$. However $\vec{u}$ is a fixed point of (\ref{3}) according to Theorem \ref{12}, thus we should have $\mathrm{d}p_1/\mathrm{d}t=0$, which is a contradiction.
\end{proof}

\begin{lemma}
	Assume that $\lambda_1$ is simple. Then $u_i>0\iff Y_i\in \bar{T}$.
	\label{41}
\end{lemma}
\begin{proof}
	Apply the Perron-Frobenius theorem to $\textbf{A}_{\bar T}$, the restriction of $\bf A$ on $\bar{T}$, and let $\vec{w}$ be its Perron eigenvector. $\textbf{w}_T$, the restriction of $\vec{w}$ on $T$ cannot be $\vec{0}$, otherwise $\lambda_1$ is an eigenvalue of $\textbf{A}_{\bar T\backslash T}$,  a contradiction. From Lemma \ref{10} we know that $\vec{w}$ is positive. Set $u_i=w_i$ if $Y_i\in \bar{T}$, and $u_j=0$ if $Y_j\notin \bar{T}$, then $\vec{u}$ is the Perron eigenvector of $\bf A$. Thus $u_i>0\iff Y_i\in \bar{T}$.\\
\end{proof}

\begin{lemma}
	[Lemma 9.8 in \cite{Janson}]Assume that $\lambda_1$ is simple and positive. Then we have almost surely $\mathrm{e}^{-\lambda_1 t}\vec{X}$$(t) \to W\vec{u}$ as $t\to\infty$, where $W$ is a nonnegative random variable, and $\mathbb{P}(W>0)>0$.
	\label{11}
\end{lemma}

\begin{lemma}
	[Lemma 9.7 (ii) and (iii) in \cite{Janson}, originated from Theorem V.7.2 in \cite{AN1972}]Assume that $\lambda_1$ is simple and positive, and $\bar{T}$ contains all phenotypes, then $W=0$ if and only if the branching process becomes essentially extinct almost surely.
	\label{23}
\end{lemma}

\begin{remark}
	\label{R2}
	Janson's paper \cite[Section 2]{Janson} has six fundamental assumptions (A1)-(A6). Assumptions (A1)-(A5) have been satisfied in this paper (regarding (A5) as ``the process is not essentially extinct at time $0$''). Assumption (A6) ``$\bar{T}$ contains all phenotypes'' is only used in Lemma \ref{23}. In fact, we will prove (in Lemma \ref{33}) that Lemma \ref{23} is still correct without Assumption (A6). Thus we can drop Assumption (A6) in the main result. Assumption (A6) implies all phenotypes should have the same exponential growth rate $\lambda_1$, and no phenotype will die out or dominate (see Section \ref{S5}), which is not necessarily satisfied in experiments. For example, in Bozic et al.'s paper \cite{Bozic}, they consider tumor cells which gradually cumulate mutations which accelerate cell growth. Since tumor cells with more accelerating mutations cannot switch back to tumor cells with less such mutations, they must have different growth rates. Also, the proportion of tumor cells with less accelerating mutations will gradually decrease to $0$, which contradicts with assumption (A6).  
\end{remark}

The following lemma is a modification of the second Borel-Cantelli lemma. We base our proof on Theorem 2.3.6 in \cite{Durrett}.
\begin{lemma}
	Consider events $B_1, B_2, \cdots, B_n, \cdots$. If for any positive integers $m<n$, we have $\mathbb{P}(\cap^n_{i=m+1}B_i^c)\le (1-\epsilon)^{n-m}$, where $0<\epsilon\le 1$, then $\mathbb{P}(\limsup_{n \rightarrow \infty}$ $B_n)=1$. In other words, almost surely $\{B_n:n\ge 1\}$ will happen infinitely often.
	\label{64}
\end{lemma}
\begin{proof}
	Let $0<M<N<\infty$. $\mathbb{P}(\cap^N_{i=M+1}B_i^c)\le(1-\epsilon)^{N-M}\rightarrow 0$ as $N\rightarrow \infty$. So $\mathbb{P}(\cup^{\infty}_{i=M+1}B_i)=1$ for all $M$, and since $\cup^{\infty}_{i=M+1}B_i\downarrow\limsup_{n\rightarrow\infty} B_n$ it follows that 
	$\mathbb{P}(\limsup_{n\rightarrow\infty} B_n)=1$.
\end{proof}

\begin{lemma}
	For almost every essentially non-extinct trajectory (according to Lemma \ref{11}, the set of such trajectories has positive probability), we can find an essentially non-extinct cell with phenotype in $T$ within finite time. If we can find such cell at time $t$, then we can find such cell at any time $\tau>t$.
	\label{54}
\end{lemma}
\begin{proof}
	If at some time $t$ all cells with phenotypes in $\hat{T}\setminus T$ die out, then at least one of the remaining cells with phenotypes in $T$ is not essentially extinct. 
	
	Otherwise, at each time $t=k$ ($k\in \mathbb{Z}^+$), there exists one cell $E_k$ with phenotype in $\hat{T}\setminus T$. (For different $k$, $E_k$ may be the same cell.)  Let $B_k$ ($k\in {\mathbb Z}^+$) be the event that during the time interval $[k,k+1)$, the cell $E_k$ produces (directly or indirectly) at least one cell with phenotype in $T$. 
	
	If $B_k$ happens, choose one such cell with phenotype in $T$ and put it in a special set $S$. Consider any two cells $F$ and $G$ in $S$, and assume $F$ is produced in the time interval $[i,i+1)$, $G$ is produced in the time interval $[j,j+1)$, and $i<j$, where $i, j\in\mathbb{Z}^+$. Then $E_j$ is the ancestor of $G$. Since $E_j$ has phenotype in $\hat{T}\setminus T$, and $F$ has phenotype in $T$, $F$ cannot be the ancestor of $E_j$. Since $E_j$ is still alive at time $t=j$, when $F$ has been produced, $E_j$ cannot be the ancestor of $F$. Thus $F$ cannot be the ancestor of $G$. Since $G$ is produced after $F$, $G$ cannot be the ancestor of $F$. In sum, one cell in $S$ cannot be the ancestor of another cell in $S$. Thus all cells in $S$ are independent. 
	
	Consider two phenotypes $Y_i$ and $Y_j$, and assume a cell with phenotype $Y_i$ can produce a cell with phenotype $X_j$ directly, namely $\mathbb{P}(d_{ij}>0)>0$. Because of Markovian property, within a time span of $1/n$, the probability for a cell with phenotype $Y_i$ to produce a cell with phenotype $Y_j$ directly is $\eta_{ij}=[1-\exp(-\alpha_i /n)]\mathbb{P}(d_{ij}>0)>0$. Let $\eta=\min_{i,j}\{\eta_{i,j}:\mathbb{P} (d_{ij}>0)>0\}$. For a cell with phenotype in $\hat{T}\setminus T$, it can produce a cell with phenotype in $T$ within $n$ steps. Thus the probability of $B_k$ is no less than $\eta^n$, regardless of what happens before time $t=k$. 
	
	Now we can use Lemma \ref{64} with $\epsilon=\eta^n$, and there will be an infinite number of cells in $S$, except for a zero-measure set of trajectories. According to Lemma \ref{11}, the probability for one cell in $S$ to become essentially extinct is less than $1$, thus the probability for all cells in $S$ to become essentially extinct is $0$, and at least one cell in $S$ is not essentially extinct, except for a zero-measure set of trajectories.
\end{proof}

\begin{lemma}
	Assume that $\lambda_1$ is simple and positive, then $W=0$ if and only if the branching process becomes essentially extinct almost surely.
	\label{33}
\end{lemma}

\begin{proof}
	$\Leftarrow$: For a trajectory $\vec{X}(\cdot)$ outside the zero-measure exclusion set of Lemma \ref{11}, assume that at some time $\tau\ge0$ (dependent on the trajectory), $X_i(\tau)=0$ for all $Y_i\in \hat{T}$. For any $Y_j\in T$, $0=\lim_{t\to \infty} \mathrm{e}^{-\lambda_1 t}X_j(t)=Wu_j$. From Lemma \ref{41}, $u_j>0$. Thus $W=0$ almost surely.\\
	
	$\Rightarrow$: Assume that $\mathbb{P}(W=0$ $\&$ the trajectory is not essentially extinct$)=P_0>0$. According to Lemma \ref{54}, we can find time $t_0>0$ large enough such that $\mathbb{P} (W=0$ $\&$ the trajectory is not essentially extinct $\&$ there exists an essentially non-extinct cell with phenotype in $T$ at time $t_0)\ge P_0/2>0$. On this set, only consider this essentially non-extinct cell and its descendants from time $t\ge t_0$, then the population is restricted on $\bar{T}$  and we can use Lemma \ref{23}. Now we have $W>0$ except for a zero-measure set of trajectories, which is a contradiction.
\end{proof}

From Lemma \ref{11} and Lemma \ref{33},we can obtain Theorem \ref{42}.

\begin{remark}
	The assumption of $\lambda_1>0$ is not too strong. If $\lambda_1<0$, then from Theorem \ref{12}, the expected populations decay to $\vec{0}$. Therefore this process will become extinct almost surely. For $\lambda_1=0$, consider an example that each cell always have exactly one child, and the child can be any phenotype with equal probability. Then the total population is fixed, and the proportions will always fluctuate, so there is no convergence \cite{Janson}.
\end{remark}

For Gupta el al's experiment, the initial cell population is very large in cancer cell lines, thus the probability of essential extinction is quite small. Therefore, the proportions will almost always tend to the same constants. This gives a satisfactory stochastic explanation of the phenotypic equilibrium phenomenon reported in \cite{Gupta}. 

The deterministic model only reflects the average behavior of many trajectories (or equivalently a large initial population). When the cell number is relatively small, the stochasticity is not negligible, and it is not very reasonable to assume the cell number changes continuously. So the stochastic model is more effective than deterministic model. That is why we also prove the same result for stochastic model.

\begin{remark}
	In Bozic et al.'s model, we divide cells in groups by their number of driver mutations. The growth rate for cells with $k$ driver mutations is $1-(1-s)^k/2$ ($s=0.004$ in \cite{Bozic}). Cells can acquire driver mutations, but cannot lose them. Also we do not consider cell death, so there is no essential extinction. If we only consider cells with no more than $n$ driver mutations, then the population of cells with exactly $n$ mutations will grow with exponential growth rate $1-(1-s)^n/2$, which is larger than that of cells with less driver mutations. So cells with $n$ driver mutations will dominate exponentially fast. Generally, the population of cells with $n$ driver mutations will grow exponentially with rate $1-(1-s)^n/2$, and as long as the next driver mutation emerges, its proportion will decay to $0$ exponentially fast. Also, we should notice that such results are valid for almost every trajectories. Here we do not consider the difference in passenger mutations, since they have no effect on cell growth, and considering them will make the Perron eigenvalue not simple. The model in \cite{Bozic,KC13} is discrete-time, but we can see from Section \ref{S6} that we can still apply our results.
\end{remark}

\section{When will one proportion tend to $0$ or $1$?}
\label{S5}
In population dynamics, we are also concerned about when one phenotype dies out or dominates. In terms of the notations in this paper, we need to consider when $P_i(t)\to 0$ or $P_i(t)\to 1$ as $t\rightarrow \infty$.\\

In this section, we will still assume that the Perron eigenvalue $\lambda_1$ of $\bf A$ is simple and positive. Then from Theorem \ref{42}, we have $(P_1(t),P_2(t),\cdots ,P_n(t))\to \vec{u}$$=(u_1,u_2,\cdots ,u_n)$ almost surely in the stochastic model. Thus we can get the following corollaries from Lemma \ref{41}.

\begin{corollary}
	$P_i(t)\to 0\iff Y_i\notin \bar{T}$.
	\label{L6}
\end{corollary}
\begin{corollary}
	$P_i(t)\to 1\iff \bar{T}=T=\{Y_i\}$.
\end{corollary}

\begin{remark}
	From Corollary \ref{L6} we can see that the sufficient and necessary condition under which no phenotype dies out, namely $\forall i, P_i(t)\nrightarrow 0$, is that $\bar{T}$ contains all phenotypes. This is just Janson's last assumption.
\end{remark}

\begin{remark}
	If we find that $P_i(t)\to 0$, $P_j(t)\nrightarrow 0$ in an experiment, then we know that the phenotype $Y_j$ will never transform to $Y_i$ in any way. If we find that $P_i(t)\to 1$, then we know that the phenotype $Y_i$ will never transform to any other phenotypes.
\end{remark}

\section{Model generalization: non-exponential lifetime}
\label{S6}
In the previous sections, we assumed that the lifetime of a cell is exponentially distributed and independent of the type and number of its descendants. However, in real biological system, the lifetime distribution should be more like lognormal, gamma, Weibull, or exponentially modified Gaussian distribution \cite{Haw,Gol}. Furthermore, the time needed for division and conversion have different distributions \cite{Gol}. In this way the process is a multitype Bellman-Harris branching process (also called age-dependent branching process) \cite{mode}, no longer Markovian.

We can use the ``device of stages'' method to approximate a non-exponential random variable with several exponential random variables \cite{Cox}. This indicates that through adding supplementary sub-phenotypes, we can simulate a non-Markovian branching process with a Markovian branching process. See the example below:

$$
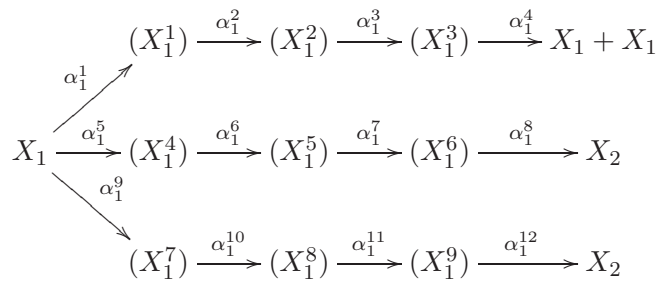
\xymatrix{
	& (Y_1^1)\ar[r]^{\alpha_1^2}&(Y_1^2) \ar[r]^{\alpha_1^3}&(Y_1^3)\ar[r]^{\alpha_1^4}&Y_1+Y_1\\
	Y_1\ar[dr]^{\alpha_1^9}\ar[ur]^{\alpha_1^1}\ar[r]^{\alpha_1^5}&(Y_1^4)\ar[r]^{\alpha_1^6}&(Y_1^5)\ar[r]^{\alpha_1^7}   &(Y_1^6)\ar[r]^{\alpha_1^8}&Y_2 \\
	&(Y_1^7)\ar[r]^{\alpha_1^{10}}    & (Y_1^8)\ar[r]^{\alpha_1^{11}}& (Y_1^9)\ar[r]^{\alpha_1^{12}}& Y_2   
}$$

Here $(Y_1^1),\cdots, (Y_1^9)$ are supplementary sub-phenotypes. We artificially assume such supplementary sub-phenotypes exist just by technical reasons. They do not have biological meanings. When we count $Y_1^i$ as $Y_1$, the process has the same distribution with the original one. If we have convergence with this new process, then we also have convergence for the original process. An $Y_1$ cell has probability $\alpha_1^1/(\alpha_1^1+\alpha_1^5+\alpha_1^9)$ to divide into $Y_1+Y_1$, and probability $(\alpha_1^5+\alpha_1^9)/(\alpha_1^1+\alpha_1^5+\alpha_1^9)$ to convert into $Y_2$. Here we set $\alpha_1^1$, $\alpha_1^5$, and $\alpha_1^9$ to be large enough while keeping their proportions, so that the time needed for the first step is ignorable (exponential random variable with expectation $1/(\alpha_1^1+\alpha_1^5+\alpha_1^9)$).

Now the time distribution for division $Y_1\to Y_1+Y_1$ is approximately $Ex(\alpha_1^2)*Ex(\alpha_1^3)*Ex(\alpha_1^4)$, where $Ex(\alpha)$ is the density function of exponential random variable with parameter $\alpha$, and $*$ means convolution. Similarly, the time distribution for conversion $Y_1\to Y_2$ is approximately $\frac{\alpha_1^5}{\alpha_1^5+\alpha_1^9}Ex(\alpha_1^6)*Ex(\alpha_1^7)*Ex(\alpha_1^8)+\frac{\alpha_1^9}{\alpha_1^5+\alpha_1^9}Ex(\alpha_1^{10})*Ex(\alpha_1^{11})*Ex(\alpha_1^{12})$.

According to \cite{Cox}, any non-negative random variable can be approximated to any accuracy by such combination of convolutions of exponential random variables. Thus we can simulate such non-Markovian branching processes to any precision with Markovian branching processes. Here the lifetime of a cell can be non-exponential, and the lifetime of a cell can depend on the type and number of its descendants.

Now we can apply Theorem \ref{42} to those sub-phenotypes. The proportion of each sub-phenotype converges to a constant. Thus the proportion of each phenotype (including all its sub-phenotypes) converges to a constant. This proves the ``phenotypic equilibrium'' phenomenon in a more realistic stochastic model. In addition, the conclusions in Section \ref{S5} are still valid.

\begin{remark}
	The most unrealistic aspect of exponential lifetime is that the density function reaches maximum at $0$, but one cell cannot divide right after its birth. For the sum of several independent exponential variables, the density function at $0$ is $0$. By the law of large numbers, the density function of the sum of $n$ independent exponential variables with parameter $n\lambda$ will have a sharp peak near $\lambda$ when $n$ is large.
\end{remark}

\begin{remark}
	The proportion convergence theorem for non-Markovian (age-dependent) branching processes can be proved directly, but under stronger conditions \cite{mode}.
\end{remark}

\section{Conclusion and discussion}
We have presented a unified stochastic model for the population dynamics with cellular phenotypic conversions. We have given the sufficient and necessary condition under which the dynamical behavior of our model can be described by an $n$-state Markov chain. In general case, we have proved that the proportions of different phenotypes will tend to constants regardless of their initial values, and we have investigated the sufficient and necessary conditions under which one phenotype will die out or dominate. We also extend our model to non-Markovian case while keeping the above conclusions valid. In this way we explain experimental phenomenon in \cite{Gupta}.

As remarked in Section \ref{S42}, we improve a limit theorem in branching processes, which may be of theoretical interests.

Our results can also apply to cancer progression models with gene mutations, such as \cite{Bozic,KC13,MK15}. 

Since the phenotypic conversions have been reported in various cellular systems, such as \textit{E.coli} \cite{OTL04} and cancer cells \cite{FK77,YQ12}, we hope that our model here could be applied as a general framework in the study of multi-phenotypic populations of cells.

With the improvement of experiment methods, we will accumulate more and more data for single cell. In single cell level, the stochasticity is not negligible anymore. Thus we should build detailed stochastic models (branching process would be a good framework) for cell population dynamics. Such models will provide new insights and predictions. In the meanwhile, stochastic process related theories should attract more attention.

There are some possible improvements about this research. First, we assume that the branching process is time homogeneous, namely the birth and death rates keep the same for all time. However, as time goes on, the cell density increases, and the birth and death rates should change \cite{ZJ}. Thus a possible improvement is to have time-dependent or density-dependent $d_{ij}$. Second, we only prove the convergence for $t\to \infty$, but in experiments we only have finite observation time. Thus it is meaningful to estimate the convergence rate. Third, we only consider finite many phenotypes. We find that there is essential difficulty to build similar theory for infinite-type branching processes. However, there have been some works considering infinite many mutation states with multitype branching processes, such as \cite{MK15}.

\section*{Acknowledgements}
We would like to thank Professor Min-Ping Qian, Da-Yue Chen, Svante Janson and anonymous reviewers for helpful advice and discussions. Y. W. would like to thank Lingxue Zhu and Mingda Zhang for a special and inspiring discussion. 


%
%
%


\begin{thebibliography}{}
	
	
	\expandafter\ifx\csname url\endcsname\relax
	\def\url#1{\texttt{#1}}\fi
	\expandafter\ifx\csname urlprefix\endcsname\relax\def\urlprefix{URL }\fi
	\expandafter\ifx\csname href\endcsname\relax
	\def\href#1#2{#2} \def\path#1{#1}\fi
	
	
	\bibitem{AW10}
	Altschuler SJ, Wu LF. Cellular heterogeneity: do differences make a
	difference? Cell 141(4) (2010) 559--563.
	\newblock \href {http://dx.doi.org/10.1016/j.cell.2010.04.033}
	{\path{doi:10.1016/j.cell.2010.04.033}}.
	
	\bibitem{KL05}
	Kussell E, Leibler S. Phenotypic diversity, population growth, and
	information in fluctuating environments. Science 309~(5743) (2005)
	2075--2078.
	\newblock \href {http://dx.doi.org/10.1126/science.1114383}
	{\path{doi:10.1126/science.1114383}}.
	
	\bibitem{dSdS13}
	dos Santos RV, da~Silva LM. The noise and the {KISS} in the cancer stem
	cells niche. J. Theor. Biol 335 (2013) 79--87.
	\newblock \href {http://dx.doi.org/10.1016/j.jtbi.2013.06.025}
	{\path{doi:10.1016/j.jtbi.2013.06.025}}.
	
	\bibitem{Gupta}
	Gupta PB, Fillmore CM, Jiang G, Shapira SD, Tao K, Kuperwasser C, et al.
	Stochastic state transitions give rise to phenotypic
	equilibrium in populations of cancer cells. Cell 146~(4) (2011) 633--644.
	\newblock \href {http://dx.doi.org/10.1016/j.cell.2011.07.026}
	{\path{doi:10.1016/j.cell.2011.07.026}}.
	
	
	
	\bibitem{ZWL13}
	Zhou D, Wu D, Li Z, Qian MP, Zhang MQ. Population dynamics of cancer
	cells with cell state conversions. Quant. Biol. 1~(3) (2013) 201--208.
	\newblock \href {http://dx.doi.org/10.1007/s40484-013-0014-2}
	{\path{doi:10.1007/s40484-013-0014-2}}.
	
	\bibitem{YQ12}
	Yang G, Quan Y, Wang W, Fu Q, Wu J, Mei T, et al. Dynamic equilibrium between
	cancer stem cells and non-stem cancer cells in human {SW}620 and {MCF}-7
	cancer cell populations. Br. J. Cancer 106~(9) (2012) 1512--1519.
	\newblock \href {http://dx.doi.org/10.1038/bjc.2012.126}
	{\path{doi:10.1038/bjc.2012.126}}.
	
	\bibitem{Jagers}
	Jagers P. The proportions of individuals of different kinds in two-type
	populations. {A} branching process problem arising in biology. J. Appl.
	Probab 6~(2) (1969) 249--260.
	\newblock \href {http://dx.doi.org/10.2307/3211996}
	{\path{doi:10.2307/3211996}}.
	
	\bibitem{KA02}
	Kimmel M, Axelrod DE. Branching Processes in Biology. Springer, New York,
	2002, pp. 103--140.
	\newblock \href {http://dx.doi.org/10.1007/b97371} {\path{doi:10.1007/b97371}}.
	
	\bibitem{YMN}
	Yakovlev AY, Mayer-Proschel M, Noble M. A stochastic model of brain cell
	differentiation in tissue culture. J. Math. Biol. 37~(1) (1998) 49--60.
	\newblock \href {http://dx.doi.org/10.1007/s002850050119}
	{\path{doi:10.1007/s002850050119}}.
	
	\bibitem{YY09}
	Yakovlev AY, Yanev NM. Relative frequencies in multitype branching
	processes. Ann. Appl. Probab. 19~(1) (2009) 1--14.
	\newblock \href {http://dx.doi.org/10.1214/08-AAP539}
	{\path{doi:10.1214/08-AAP539}}.
	
	\bibitem{YY10}
	Yakovlev AY, Yanev NM. Limiting distributions for multitype branching
	processes. Stoch. Anal. Appl. 28~(6) (2010) 1040--1060.
	\newblock \href {http://dx.doi.org/10.1080/07362994.2010.515486}
	{\path{doi:10.1080/07362994.2010.515486}}.
	
	\bibitem{Murray}
	Murray JD. Mathematical Biology I: An Introduction. Springer-Verlag, Berlin,
	Heidelberg, 2001, pp. 1--6, 101--105.
	\newblock \href {http://dx.doi.org/10.1007/b98868} {\path{doi:10.1007/b98868}}.
	
	\bibitem{Clayton}
	Clayton E, Doup{\'e} PD, Klein AM, Winton DJ, Simons BD, 
	Jones PH. A single type of progenitor cell maintains normal epidermis. Nature
	446~(7132) (2007) 185--189.
	\newblock \href {http://dx.doi.org/10.1038/nature05574}
	{\path{doi:10.1038/nature05574}}.
	
	\bibitem{HW00}
	Hanahan, D, Weinberg, RA. The hallmarks of cancer. Cell, 100(1), (2000) 57--70.
	\newblock \href {http://dx.doi.org/10.1016/S0092-8674(00)81683-9}
	{\path{doi:10.1016/S0092-8674(00)81683-9}}.
	
	
	\bibitem{Bozic}
	Bozic I, Antal T, Ohtsuki H, Carter H, Kim D, Chen S, et al. Accumulation of driver and passenger mutations during tumor progression. Proc. Natl. Acad. Sci. 2010 Oct 26;107(43):18545-50.
	\newblock \href {http://dx.doi.org/10.1073/pnas.1010978107}
	{\path{doi:10.1073/pnas.1010978107}}.
	
	\bibitem{KC13}
	Kimmel M, Corey S. Stochastic hypothesis of transition from inborn neutropenia to AML: interactions of cell population dynamics and population genetics. Front. Oncol. 2013 Apr 29;3(89.10):3389.
	\newblock \href {http://dx.doi.org/10.3389/fonc.2013.00089}
	{\path{doi:10.3389/fonc.2013.00089}}.
	
	\bibitem{MK15}
	McDonald TO, Kimmel M. A multitype infinite-allele branching process with applications to cancer evolution. J. Appl. Probab. 2015;52(3):864-76.
	\newblock \href {http://dx.doi.org/10.1239/jap/1445543852}
	{\path{doi:10.1239/jap/1445543852}}.
	
	
	\bibitem{KS67}
	Kesten H, Stigum BP. Limit theorems for decomposable multi-dimensional
	{G}alton-{W}atson processes. J. Math. Anal. Appl. 17~(2) (1967) 309--338.
	\newblock \href {http://dx.doi.org/10.1016/0022-247X(67)90155-2}
	{\path{doi:10.1016/0022-247X(67)90155-2}}.
	
	\bibitem{AN1972}
	Athreya KB, Ney PE. Branching Processes. Springer-Verlag, Berlin, 1972,
	pp. 199--206.
	\newblock \href {http://dx.doi.org/10.1007/978-3-642-65371-1}
	{\path{doi:10.1007/978-3-642-65371-1}}.
	
	\bibitem{Janson}
	Janson S. Functional limit theorems for multitype branching processes and
	generalized {P}{\'o}lya urns. Stoch. Process. Appl. 110~(2) (2004) 177--245.
	\newblock \href {http://dx.doi.org/10.1016/j.spa.2003.12.002}
	{\path{doi:10.1016/j.spa.2003.12.002}}.
	
	\bibitem{ZJ}
	Zwietering MH, Jongenburger I, Rombouts FM, van't Riet K. Modeling of
	the bacterial growth curve. Appl. Environ. Microbiol. 56~(6) (1990)
	1875--1881.
	
	\bibitem{PCS}
	Patrawala L, Calhoun T, Schneider-Broussard R, Zhou J, Claypool K, 
	Tang DG. Side population is enriched in tumorigenic, stem-like cancer cells,
	whereas {ABCG}2$^+$ and {ABCG}2$^-$ cancer cells are similarly tumorigenic.
	Cancer Res. 65~(14) (2005) 6207--6219.
	\newblock \href {http://dx.doi.org/10.1158/0008-5472.CAN-05-0592}
	{\path{doi:10.1158/0008-5472.CAN-05-0592}}.
	
	\bibitem{FK08}
	Fillmore CM, Kuperwasser C. Human breast cancer cell lines contain
	stem-like cells that self-renew, give rise to phenotypically diverse progeny
	and survive chemotherapy. Breast Cancer Res. 10~(2) (2008) R25.
	\newblock \href {http://dx.doi.org/10.1186/bcr1982}
	{\path{doi:10.1186/bcr1982}}.
	
	\bibitem{Seneta}
	Seneta E. Non-negative Matrices and Markov Chains. 2nd Edition. Springer, New
	York, 1981, p.~22.
	\newblock \href {http://dx.doi.org/10.1007/0-387-32792-4}
	{\path{doi:10.1007/0-387-32792-4}}.
	
	\bibitem{Karlin}
	Karlin S. Taylor HM, A First Course in Stochastic Processes. 2nd Edition.
	Academic Press, New York, 1975, pp. 547--551.
	
	\bibitem{ZWW}
	Zhou D, Wang Y, Wu B. A multi-phenotypic cancer model with cell plasticity.
	J. Theor. Biol. 357 (2014) 35--45.
	\newblock \href {http://dx.doi.org/10.1016/j.jtbi.2014.04.039}
	{\path{doi:10.1016/j.jtbi.2014.04.039}}.
	
	\bibitem{Ath68}
	Athreya KB. Some results on multitype continuous time {M}arkov branching
	processes. Ann. Math. Stat. 39(2) (1968) 347--357.
	\newblock \href {http://dx.doi.org/10.1214/aoms/1177698395}
	{\path{doi:10.1214/aoms/1177698395}}.
	
	\bibitem{Smythe}
	Smythe RT. Central limit theorems for urn models. Stoch. Process. Appl.
	65~(1) (1996) 115--137.
	\newblock \href {http://dx.doi.org/10.1016/S0304-4149(96)00094-4}
	{\path{doi:10.1016/S0304-4149(96)00094-4}}.
	
	\bibitem{Norris}
	Norris JR. Markov Chains. Cambridge University Press, Cambridge, 1997, pp. 11, 122.
	
	\bibitem{Durrett}
	Durrett R. Probability: Theory and Examples. 4th Edition. Cambridge University
	Press, Cambridge, 2010, pp. 58--59.
	\newblock \href {http://dx.doi.org/10.1017/CBO9780511779398}
	{\path{doi:10.1017/CBO9780511779398}}.
	
	\bibitem{Haw}
	Hawkins ED, Turner ML, Dowling MR, van Gend C, Hodgkin PD. A model
	of immune regulation as a consequence of randomized lymphocyte division and
	death times. Proc. Natl. Acad. Sci. 104 (2007) 5032--5037.
	\newblock \href {http://dx.doi.org/10.1073/pnas.0700026104}
	{\path{doi:10.1073/pnas.0700026104}}.
	
	\bibitem{Gol}
	Golubev A. Exponentially modified gaussian ({EMG}) relevance to distributions
	related to cell proliferation and differentiation. J. Theor. Biol. 262 (2010)
	257--266.
	\newblock \href {http://dx.doi.org/10.1016/j.jtbi.2009.10.005}
	{\path{doi:10.1016/j.jtbi.2009.10.005}}.
	
	\bibitem{mode}
	Mode CJ. Multitype Branching Processes: Theory and Applications. Vol.~34,
	American Elsevier Pub. Co., New York, 1971, pp. 138--145.
	
	\bibitem{Cox}
	Cox DR, Miller HD. The Theory of Stochastic Processes. Methuen \& Co.
	Ltd., London, 1965, pp. 257--262.
	
	\bibitem{OTL04}
	Ozbudak EM, Thattai M, Lim HN, Shraiman BI, van Oudenaarden A.
	Multistability in the lactose utilization network of \textit{{E}scherichia
		coli}. Nature 427~(6976) (2004) 737--740.
	\newblock \href {http://dx.doi.org/10.1038/nature02298}
	{\path{doi:10.1038/nature02298}}.
	
	\bibitem{FK77}
	Fidler IJ, Kripke ML. Metastasis results from preexisting variant cells
	within a malignant tumor. Science 197~(4306) (1977) 893--895.
	\newblock \href {http://dx.doi.org/10.1126/science.887927}
	{\path{doi:10.1126/science.887927}}.
	
	
	
	
\end{thebibliography}
\end{document}